\documentclass[twoside,a4paper,12pt]{article}
\usepackage[latin1]{inputenc} 
 \usepackage[T1]{fontenc}      
\usepackage{geometry}         

\RequirePackage{textcomp}
\RequirePackage{amsthm}
\RequirePackage{amsmath}
\RequirePackage{bm}
\RequirePackage{amssymb}
\RequirePackage{upref}
\RequirePackage{amsthm}
\RequirePackage{enumerate}
\RequirePackage{pb-diagram}
\RequirePackage{amsfonts}
\RequirePackage[mathscr]{eucal}
\RequirePackage{verbatim}
\RequirePackage{xr}
\newcommand{\al}{\alpha}
\newcommand{\bet}{\beta}

\newcommand{\de}{\delta }

\newcommand{\bo}{\mathbf}

\newcommand{\demi}{\frac{1}{2}}

\newcommand{\abs}[1]{\lvert#1\rvert}

\newcommand{\norm}[1]{\lVert#1\rVert}

\usepackage{xspace}

\newtheorem{thm}{Theorem}[subsection]
\newtheorem{lem}{Lemma}[subsection]

\begin{document}
\title{Geometric uniqueness for non-vacuum Einstein equations and applications :}
\author{David Parlongue \\
  \texttt{david.parlongue@polytechnique.org}}

\maketitle

\abstract{We prove in this note that local geometric uniqueness holds true without loss of regularity for Einstein equations coupled with a large class of matter models. We thus extend the Planchon-Rodnianski uniqueness theorem for vacuum spacetimes. In a second part of this note, we investigate the question of local regularity of spacetimes under geometric bounds. }

\section{Introduction :}
Consider  Einstein field equations on an oriented Lorentzian manifold ${\bf (M,g)}$ 
\begin{eqnarray}
{\bo R}_{\al \bet }-\demi {\bo R} {\bo g}_{\al \bet }&=& {\bo T}_{\al \bet}(\Psi) \\
F(\Psi, {\bf g})&=&0
\end{eqnarray}
where $\Psi$ denotes the matter field which is supposed to be a tensor of order (p,q) or a collection of tensors, $F$ denotes the matter equations, ${\bf T}$ the energy-momentum tensor. 

\vspace{3 mm}
We suppose that ${\bf (M,g)}$ is globally hyperbolic with respect to a riemannian hypersurface $\Sigma_0$ and denote by ${\bf D}$ the covariant derivative associated to ${\bf g}$. ${\bf M}$ is globally diffeomophic to $\Sigma_0 \times \mathbb{R}$. Denote by $\nabla$ the covariant derivative relative to $g_{\Sigma_0}$ and $\partial :=(\partial_t,\nabla)$. 

Consider a tensorfield  $P$ on ${\bf M}$. For $k \in \mathbb{N}$, we denote by :
\begin{equation*}
\norm{P}^2_{X^k_T} = \sup_{t \in ]0,T[} \sum_{0\leq |j|\leq k}\int_{\Sigma_t} \abs{\partial^j P}^2 d\mu_{g}
\end{equation*}

    The $X^k_T$ norms can be extended to fractional indices $s \in \mathbb{R}$. We will suppose in this note that $\Sigma_0$ is compact but similar results are true if $\Sigma_0$ admits a locally finite $C^1$-covering.

\vspace{3 mm}

An abstract initial data set for the Cauchy problem for Einstein equations consists in $(\Sigma_0, g_0,k_0,\rho,j)$ where $(\Sigma_0,g_0)$ is a three dimensional Riemannian manifold, $k_0$ a two times covariant tensorfield on $\Sigma_0$, $\rho$ a scalar, $j$ a vectorfield on $\Sigma_0$ satisfying the so-called constraint equations :
\begin{eqnarray}
R- \abs{k}^2 + tr(k)^2 &=& 2 \rho \\
\nabla^i k_{il} - \nabla_l tr(k)&=& j
\end{eqnarray}

Once a matter model satisfying some reasonable assumptions (see below) has been given, the first question is local-in-time well-posedness of the Cauchy problem for the Einstein-matter equations with initial data of a given regularity (typically $H^s(\Sigma_0)$) in a given gauge. However a standard pathology in the mathematical relativity literature is that local uniqueness up to a diffeomorphism which we will call geometric uniqueness requires more regularity on the initial data than the existence theorem in a given gauge. This phenomenon is a consequence of the following remark. 

If we consider a solution of Einstein equations in $X^s_T$ in a given gauge, we can construct in the neighborhood of $\Sigma_0$ a system of wave coordinates $(x^0,x^1,x^2,x^3)$, id est :
\begin{equation}
\square_{\bf g}x^{\al}=0
\end{equation}
Standard techniques (see lemma 3.0.1) imply that $x_{\al} \in X^s$ and thus ${\bf g}$ in the new coordinates lie only in $X^{s-1}$. It is well known that in Riemannian geometry, harmonic coordinates achieve optimal regularity for the metric (DeTurck-Kazdan theorem \cite{DeTurckKazdan}). It is not clear however by energy techniques that a lorentzian metric has optimal regularity in harmonic coordinates.

This rough analysis does not take however advantage of Einstein equations. F.Planchon and I.Rodnianski have remarked in \cite{PR} that in fact, the wave coordinates lie in $X^{s+1}$ and thus geometric uniqueness for the vacuum equations holds true without additional regularity. Their proof is based on the fact that the Riemann curvature tensor of a vacuum Einstein manifold is divergence free. However this property is not shared by other matter models, in general :
\begin{equation}
{\bo C}_{\bet \gamma \delta}= {\bo D}^{\al} {\bo R}_{\al \bet \gamma \delta} = {\bo D}_{\gamma }{\bo {Ric}}_{\bet \delta}- {\bo D}_{\delta }{\bo {Ric}}_{\gamma \bet}
\end{equation}

It is thus not clear whether the property of geometric uniqueness can be extended to a more general setting. 
Note also that a consequence of this loss of regularity is that the  existence of maximal globally hyperbolic development requires also one additional degree of Sobolev regularity in comparison to local existence in wave gauge (also called de Donder gauge).

We will thus study in this paper the well-posedness for the Cauchy problem for the geometric linear wave equation {\bf (L) }:
\begin{eqnarray*}
        {\bf D}_\al {\bf D}^\al \phi = \square_{\bf g} \phi &= &{\bf g}^{\al \bet}\partial_\al \partial_\bet \phi - {\bf \Gamma}^\al \partial_\al \phi   =0 \\
       \phi(0,.)&=&\phi_0 \in  H^r(\Sigma_0)\\
        \partial_t \phi(0,.)&=& \phi_1  \in H^{r-1}(\Sigma_0)
\end{eqnarray*}
where ${\bf \Gamma}^\al= {\bf g}^{\de \gamma}{\bf \Gamma}^\al_{\de \gamma }$.
We prove :
\begin{thm}
Given a generic Lorentzian metric ${\bf g} \in X^s_T$, the Cauchy problem {\bf (L)} is well-posed for $1\leq r \leq s$ whereas under physically reasonable assumptions (see below), the metric of a spacetime of general relativity is such that  {\bf (L)} is well-posed for $1 \leq r \leq s+1$.
\end{thm}

As a consequence of the previous theorem, Lorentzian metrics of physically reasonable spacetimes have optimal regularity in de Donder gauge in the sense that if a metric lies in $X^s_T$ in a given gauge, it is also lies (locally) in $X^s_T$ in de Donder gauge. A natural question is if one can obtain optimal regularity in de Donder gauge under bounds on geometric quantities only, this is the subject of the second part of this note.

\section{General setting :}

Following \cite{HE}, we suppose that the matter equation with given background metric satisfy the following existence, uniqueness and stability properties :

\begin{itemize}
\item ${\bf T_{\al \bet}} $ is polynomial in $\Psi$ and ${\bf D} \Psi$ and ${\bf g}$
\item for a given ${\bf g}\in X^s_T$ and $[\Psi_0]=(\Psi,\partial_t \Psi)\in  H^s(\Sigma_0)\times  H^{s-1}(\Sigma_0)$, $F(\Psi, {\bf g})=0$ has a unique solution in $X^s_{T'}$, where $0<T'\leq T$ depends continuously on the $X^s \times H^s$ norm of $({\bf g}, [\Psi_0])$. Denote by $\Psi_{{\bf g}, \Psi_0}$ this solution. We suppose moreover that for $S\subset \Sigma_0$ a smooth submanifold of $\Sigma_0$,  the value of $\Psi_{{\bf g}, \Psi_0}$ on $\mathcal{D}^{+}(S)$ the future domain of dependance of $S$ in ${\bf M}$  depends only on the initial data restricted to $S$.

\item We suppose that the following stability conditions are satisfied : for ${\bo g} \in X^s_T$, there exists $C({\bo g}, \Psi_0)>0$ and $D({\bo g, \Psi_0})>0$ s.t. for any ${\bo g'}$ and $\Psi'_0$, $\norm{{\bf {g- g'}}}_{X^s_T}+\norm{ [ \Psi_0 - \Psi'_0  ]} _{H^s} \leq C$ : 
\begin{equation}
\norm{{\bf \Psi}({\bo g} ,\Psi_0)-{\bf \Psi}({\bo {g'}} ,\Psi'_0)}_{X^s_{T'}} \leq D(\norm{{\bo {g- g'}}}_{X^s_T} + \norm{ [ \Psi_0 - \Psi'_0  ]} _{H^s}) 
\end{equation}
\end{itemize}

We call such a matter model $s$-admissible. Note that we allow the matter equation to have its own breakdown mechanism ($T'<T$). These properties arise as a consequence of physical assumptions (see \cite{HE}) on the matter model. 
\begin{thm}
Consider  Einstein equations coupled with an $s$-admissible  matter model for $s>5/2$ then existence and uniqueness hold true for the reduced system in wave coordinates and data $(g,k,[\Psi_0])\in H^s(\Sigma_0)\times  H^{s-1}(\Sigma_0) \times (H^s(\Sigma_0)\times  H^{s-1}(\Sigma_0))$ satisfying the constraint equations. The solution lies in $X^s_T$ and depends continuously on the initial data.
\end{thm}
\begin{proof}
The reduced system in wave coordinates is of the following form :
\begin{eqnarray}
\square_{\bf g}{\bf g}_{\al \bet}&=& \mathcal{N}_{\al \bet}({\bf g}, \partial{\bf g}) -2 {\bf T}_{\al \bet}(\Psi)+ {\bf T}(\Psi){\bf g}_{\al \bet} \\
F(g,\Psi)&=&0
\end{eqnarray}

where $\mathcal{N}_{\al \bet}$ is quadratic in $\partial{\bf g}$. The existence and uniqueness theorem for this reduced system can be proved by standard energy methods (see \cite{HE} for instance).
\end{proof}

 Remark however that a solution $({\bf g},\Psi) \in X^s_T$ of Einstein equations coupled with a $s$-admissible matter model is such that ${\bf {DRic}}\in L^{\infty}_T H^{s-2}(\Sigma)$. We will see that this implies that the solution of the wave equation with initial data in $H^{s+1}(\Sigma_0) \times H^s(\Sigma_0)$ are in $X^{s+1}_T$ and thus :

\begin{thm}
Consider Einstein equations coupled with a $s$-admissible matter model for $s>5/2$. Consider an initial data set satisfying the constraint equations $(g,k,[\Psi_0])\in H^s(\Sigma_0)\times  H^{s-1}(\Sigma_0) \times (H^s(\Sigma_0)\times  H^{s-1}(\Sigma_0))$. Consider two Cauchy developments $({\bf M}, {\bf g}, {\bf \Psi})$ and $({\bf M'}, {\bf g'}, {\bf \Psi}')$ of $(\Sigma_0,g,k,[\Psi_0])$ Then there exists isometric neighborhoods $\mathcal{U}$ and $\mathcal{U}'$  of the embedding of $\Sigma_0$ in ${\bf M}$ and ${\bf M}'$.
\end{thm}

The fact that the Ricci tensor is more regular then the Riemann curvature tensor has no reason to be true for a generic Lorentzian metric. However considering equation (8), we remark that for a given ${\bf T}_{\al \bet}\in X^s_T$ we can only hope to prove by energy methods that ${\bf g}$ to be in $X^{s+1}_T$ and not in $X^{s+2}_T$ as one could expect. That's why standard admissibility conditions required in the literature to prove existence and uniqueness for the reduced system in wave coordinates imply in fact that Ricci tensor for these matter models are more regular that their curvature tensor of one degree of regularity in spatial Sobolev norms. We will take advantage of this gap of regularity between the conformal part of the curvature tensor and the Ricci part.
 
   Examples of admissible matter models include models for which the matter equation is a quasi-linear second order hyperbolic system or first order symmetric hyperbolic system provided that the null cone for the matter equation lye within or coincide with the light cone (see \cite{HE}). 
   
   Let us recall some examples of models entering these categories :

- {\bf Einstein scalar-field :}

In this case $\psi$ is a scalar-field :
\begin{eqnarray*}
{\bf R}_{\al \bet}&=& {\bf D}_{\al } \psi {\bf D}_{\bet} \psi \\
\square_{\bf g} \psi&=& 0
\end{eqnarray*} 

This model is $s$-admissible for $s>5/2$ (see \cite{CB} chap. VI).

- {\bf Einstein Maxwell :}
\begin{eqnarray*}
{\bf R}_{\al \bet}-\demi {\bf R} {\bf g}_{\al \bet}&=& {\bf F}_{\al}^{\hspace{2mm} \nu }{\bf F}_{\bet \nu } - \frac{1}{4}{\bf g}_{\al \bet} {\bf F}^{\mu \nu } {\bf F}_{\mu \nu }
\end{eqnarray*}
and ${\bf F}$ satisfies Maxwell equations. This model is also $s$-admissible for $s>5/2$ (see \cite{CB} chap. VI).

- {\bf Einstein Euler :}
\begin{eqnarray*}
{\bf R}_{\al \bet}-\demi {\bf R} {\bf g}_{\al \bet}&=& (p+\rho) {\bf u}_{\al}{\bf u}_{\bet} + p {\bf g}_{\al \bet}
\end{eqnarray*}

where ${\bf u}_{\al}{\bf u}^{\al}=-1$ denotes the flow vector, $\rho$ the energy density and $p$ the pressure density. In addition, an equation of state has to be given. Under suitable assumptions on this equation, Einstein-Euler equations reduce to a symmetric hyperbolic system. Energy methods imply that the model is $s$-admissible (see \cite{CB} chap. IX) provided that the speed of sound is smaller than the speed of light and $s>7/2$.

   An obvious change in the functional setting enables to include Boltzmann and Vlasov models (see \cite{CB} chap. X). The setting of this note thus includes classical matter theories (at least for large $s$) and enables to fill the gap between the regularity needed for the reduced system in wave gauge and the local geometric uniqueness  result. Note that this short paper does not cover the more subtle question of low regularity solutions for Einstein equations ($2<s<5/2$) which will be covered by a work in preparation (\cite{PPR}). In particular  lemma (3.0.1)  is true by standard methods only for $s>5/2$.
  
\section{Geometric uniqueness :}
According to the previous arguments, theorem (2.0.3) reduces to :
\begin{thm}
 Let $({\bf g}, \Psi) \in X^s _T\times X^s_T$ be a solution of the system of equations (1),(2). Suppose that the matter model is $s$-admissible for an $s>5/2$. Let us consider a solution of  $\square_{\bf g} \phi =0$ with initial data $(\phi_0, \phi_1) \in H^{r}(\Sigma_0)\times H^{r-1}(\Sigma_0)$ for $1 \leq r \leq s+1$, then $\phi \in X^{r}_T$.
\end{thm}
\begin{proof}
Note that by standard energy techniques, we have :
\begin{lem}
Consider a solution of the geometric wave equation $\square_{\bo g} \psi = F$, with initial data $(\psi_0, \partial_t \psi_0) \in H^{s'}(\Sigma_0)\times H^{s'-1}(\Sigma_0)$,  ${\bo g} \in X^s_T$ and $F\in L^1H^{s'-1}(\Sigma)$ with $s>5/2$ and $1 \leq s'\leq s$, then $\psi \in X^{s'}_T$ and :
$$
\| \psi \|_{X^{s'}_T} \lesssim \|(\psi_0, \partial_t \psi_0) \|_{ H^{s'}(\Sigma_0)\times H^{s'-1}(\Sigma_0)} + \|F\|_{ L^1H^{s'-1}}
$$
\end{lem}

 Now note that we have the following iterated wave equations for $\phi$ :
\begin{eqnarray}
\square_{\bo g}\phi&=&0 \\
\square_{\bo g}{\bo D}_\al \phi &= & {\bo {R}} _{\al}^{\hspace{2mm}\bet} {\bo D}_\bet \phi \\
\square_{\bo g}{\bo D}^2_{\al \bet}\phi &=& {\bo D}_{\al} ( {\bo {R}} _{\bet}^{\hspace{2mm}\mu} {\bo D}_\mu \phi) + {\bo {R}}^{\mu}_{\hspace{2mm}\al} {\bo D}^2_{\mu \bet} \phi + {\bo {R}}_{\betÊ\hspace{3mm} \al}^{\hspace{2mm}\lambda \mu}{\bo D}^2_{\lambda \mu} \phi \\
\nonumber && + {\bo C}_{\al \bet \lambda}  {\bo D}^{\lambda} \phi + {\bo R}_{\mu \al \bet \lambda} {\bo D}^{\lambda \mu} \phi \\
&=& ({\bo D}_{\al}  {\bo {R}} _{\bet \mu} + {\bo D}_{\bet}  {\bo {R}}_{\al \mu}  - {\bo D}_{\mu}  {\bo {R}} _{\al \bet}  ) \cdot {\bo D}^\mu \phi \\ \nonumber
&&+ {\bo R}_\al ^{\hspace{2mm}\gamma} {\bo D}^2 _{\gamma \beta} \phi + {\bo R}_\bet ^{\hspace{2mm}\gamma} {\bo D}^2 _{\gamma \al } \phi + 2 {\bo R}^{\rho \hspace{4mm} \sigma }_{\hspace{2mm} \al \bet} {\bo D}^2 _{\rho \sigma } \phi
\end{eqnarray}

Note that in $(13)$, the only terms involving three derivatives of ${\bo g}$ lie in ${\bo A}$ where ${\bo A}_{\al \bet \delta}= {\bo D_{\al}}{\bo {R}}_{\bet \delta} + {\bo D_{\bet}}{\bo {R}}_{\al \delta} -  {\bo D_{\delta}}{\bo {R}}_{\al \bet}$. Now as ${\bf g}\in X^s_T$, ${\bf R}\in L^{\infty}H^{s-2}(\Sigma)$  and as the matter-model is $s$-admissible, ${\bf A}\in L^{\infty}H^{s-2}(\Sigma)$. Applying lemma (2.0.1) to (10), (11), (12), we obtain that $\phi \in X^{s+1}$. 
\end{proof}

\section{Regularity under curvature bounds :}

Note that this result can be rephrased in terms of local regularity under geometric bounds. Consider an admissible time-orientable spacetime ${\bf (M,g,\Psi)}$, there exists a smooth timelike unit future-directed vectorfield ${\bf T}$. Denote by ${\bf h}$ the Riemannian metric :
\begin{equation}
{\bf h}_{\al \bet} = {\bf g}_{\al \bet} + 2 {\bf T}_\al {\bf T}_\bet
\end{equation}

and by $^{\bf (T)}{\pi}$ the deformation tensor of ${\bf T}$ :
\begin{equation}
^{\bf (T)}{\pi}_{\al \bet}= {\bf D}_{\al} {\bf T}_{\bet} +{\bf D}_{\bet} {\bf T}_{\al} 
\end{equation}

For a point $p \in {\bf M}$, we denote by $B(p,a)$ the geodesic ball about $p$ of ${\bf h}$-radius $a$ and $inj(p)$ the radius of injectivity of $p$ in ${\bf (M,h)}$. Consider a ball $B(p,r)$ foliated by a time function $t$ s.t. ${ \bf T }$ is normal to the leaves $\{t=c\}$, we consider the rescaled $X^s(B(p,r))$ norm. For $ k \in \mathbb{N}$, define :

\begin{equation*}
\norm{f}^2_{X^k_T} = \sup_{t \in ]0,T[} \sum_{0\leq |j|\leq k} r^{2 |j|-3}  \int_{\Sigma_t} \abs{\partial^j f}^2 d\mu_{g}
\end{equation*}

These norms can be extended to non negative real indices $s$. As a consequence of the admissibility condition, there exists $(n,m) \in \mathbb{N}^2$ s.t. :
\begin{equation} 
\norm{{\bf Ric(\Psi)}}_{X^{s-1}(B(p,r))} \lesssim r^n\norm{{\bf \Psi}}^m_{X^s(B(p,r))}
\end{equation}

We then have the following theorem :
\begin{thm}
Consider an $s$-admissible spacetime ${\bf (M,g,\Psi,T)}$ for $s>5/2$, $ p \in {\bf M}$, $r>0$, $\al>0$ s.t. :
\begin{eqnarray*}
inj(p)&>&r \\
r^2\norm{ {\bf R}}_{X^{s-2}(B(p,r))} + r^{n+2}\norm{{\bf \Psi}}^m_{X^{s}(B(p,r)} + r \norm{^{\bf (T)}{\pi}}_{X^{s-1}(B(p,r))}&\leq& \al 
\end{eqnarray*}
then there exists $r'=rc(\al)>0$ s.t. $B(p,r')$ can be covered by wave coordinates $(t_w, x^1,x^2,x^3)$ s.t. in these coordinates :
\begin{equation*}
\norm{{\bf g}}_{X_w^s(B(p,r'))} \leq C(\al)
\end{equation*}
where $X_w^s(B(p,r'))$ denotes the $X^s(B(p,r'))$ space constructed from the wave coordinates 
$(t_w, x^1,x^2,x^3)$.
\end{thm}
\begin{proof}
Denote by $S_p$ the $t$ slice s.t. $p\in S_p$ and $g_p$ the induced Riemannian  metric on $\Sigma_p:= S_p \cap B(p,r)$.  Its curvature satisfies :
\begin{eqnarray}
r^2 \norm{R}_{H^{s-2}(\Sigma_p)}& \lesssim& r^2 \norm{{\bf R}}_{X^{s-2}} + r^2 \norm{^{\bf (T)}{\pi}^2}_{X^{s-2}}   + r^2 \norm{{\bf D} ^{\bf (T)}{\pi}}_{X^{s-2}} \\
& \lesssim & D(\al)
\end{eqnarray}
There exists $0<c(\al)<1$ s.t. one can construct  $H^{s+1}$ coordinates $(y^1,y^2,y^3)$ on $S_p\cap B(p,cr)$ s.t. in these coordinates :
\begin{equation}
\norm{g_{p..}}_{H^s(\Sigma_p)} +\norm{k_{p..}}_{H^{s-1}(\Sigma_p)} \lesssim D(\al)
\end{equation}
 Using energy estimates applied to (10), (11), (12), on can construct wave coordinates $x^\mu=(t_w, x^1,x^2,x^3)$ in a neighborhood of $S_p \cap B(p,cr)$ with initial data $x^i=y^i$ and $t_w=0$ on $\Sigma_p$, $\partial_t t_w=1$ and $\partial_t x^i=0$ on $\Sigma_p$. The coordinates cover a ball $B(p,r')$ with $r'\geq c'(\al)r$ and are such that :
\begin{equation}
\norm{x^\mu}_{X^{s+1}(B(p,c'r))} \lesssim D(\al)
\end{equation}

in the new coordinates :
\begin{equation}
r^2\norm{ {\bf R}}_{X_w^{s-2}(B(pc',r))} + r^2\norm{ {\bf Ric(\Psi)}}_{X_w^{s-1}(B(p,c'r))} \lesssim D(\al)
\end{equation}

now as the new coordinates are wave coordinates :
\begin{equation}
\square_{\bf g}{\bf g}_{\al \bet}= \mathcal{N}_{\al \bet}({\bf g}, \partial{\bf g}) -2 {\bf Ric}_{\al \bet}(\Psi)
\end{equation}
and $({\bf g},\partial_t{\bf g} ) \in H^{s}\times H^{s-1}$ on $\Sigma_p$ with :
\begin{equation}
\norm{({\bf g},\partial_t{\bf g} )}_{H^{s}(\Sigma_p)\times H^{s-1}(\Sigma_p)} \lesssim D(\al)
\end{equation}

energy estimates then give :

\begin{equation}
\norm{{\bf g}}_{X_w^s(B(p,cr))} \lesssim D(\al)
\end{equation}

\end{proof}

In the case of vacuum spacetimes, the metric is known to have optimal regularity under pointwise bound on the full curvature tensor in CMC-spatially harmonic coordinates (see \cite{Lefloch1}, \cite{Lefloch2}), for earlier work in time normal spatially harmonic coordinates see \cite{Anderson}. The previous theorem in de Donder gauge is complementary to the regularity result in CMCSH gauge. The main difference between these two approaches of regularity of Lorentzian metrics lye of course in the techniques used to obtain regularity of the components of the metrics : if regularity in CMCSH gauge comes from elliptic estimates on the $t$ slices, the approach presented here rely on hyperbolic estimates. 

As a consequence, there exists a natural notion of wave radius dual to the notion of harmonic radius central in the theory of convergence of Riemannian manifolds and the corresponding notion of norm similar to the one introduced by Petersen (see \cite{Petersen}). As an example of this parallel, let us state a precompactness result.

\section{A precompactness result :}

Consider an admissible spacetime ${\bf (M,g,\Psi,T)}$ for a given admissible matter model denoted by $\mathcal{F}$. If $inj(p)$ is uniformly bounded by below on ${\bf M}$, denote by $\norm{{\bf (M,g,\Psi,T)}}_{r,s}$ the smallest $\al$ (if finite) s.t. the conditions of theorem 4.0.5 holds true for every $p \in {\bf M}$. Denote by $\mathcal{X}(r,s,\al; \mathcal{F})$ the class of complete admissible spacetimes satisfying the Einstein equations coupled with the matter model $\mathcal{F}$ and such that $\norm{{\bf (M,g,\Psi,T)}}_{r,s} \leq \al$.

The following precompactness result can be proved following step by step the corresponding result of  Riemannian geometry (see \cite{Petersen}) :
\begin{thm}
The class $\mathcal{X}(r,s,\al; \mathcal{F})$ is precompact for the pointed $C^{1,\beta}$ topology for $1+ \beta < s-3/2$.
\end{thm}

It would be of interest to study to what extent the parallel between Cheeger-Gromov theory and a theory of convergence for admissible spacetimes could be developed and if such results could have any application to study stability problems in general relativity. The reader will note that Cheeger-Gromov theory has been applied by M.T.Anderson to study long-time evolution of the CMC problem (see \cite{Anderson2}), the approach here would be however different, instead of studying convergence of the Riemannian manifolds $\{ \tau=c \}$ as $\tau \rightarrow 0$ on a given expanding CMC spacetime, it would allow to study convergence of a sequence of spacetimes.

\newpage

\bibliographystyle{plain}

\end{document}